\documentclass{article}

\usepackage[a4paper, total={6in, 8in}]{geometry}

\usepackage{amsmath,amsthm,amssymb,amsfonts}
\usepackage{algorithm}
\usepackage{algpseudocode}

\newcommand{\OPT}{\mathsf{OPT}}

\newtheorem{lemma}{Lemma}
\newtheorem{theorem}{Theorem}

\newtheorem{corollary}{Corollary}

\title{Two 6-approximation Algorithms for the Stochastic Score Classification Problem\footnote{This is an independent work of the recent submission of Plank and Schewior~\cite{https://doi.org/10.48550/arxiv.2211.14082}, in which they provide a 5.8-approximation with similar ideas and a 6-approximation which is the second algorithm in this work. However, the analyses in the two papers are quite different. I am grateful to Lisa Hellerstein, Devorah Kletenik, R. Teal Witter for valuable discussions and suggestions.}}
\author{Naifeng Liu\footnote{CUNY Graduate Center, New York NY 10016, USA}}

\date{}

\begin{document}
\maketitle

\begin{abstract}
    We study the arbitrary cost case of the unweighted Stochastic Score Classification (SSClass) problem.
    % which is also the Stochastic Boolean Function Evaluation problem for symmetric Boolean functions
    We show two constant approximation algorithms and both algorithms are 6-approximation non-adaptive algorithms with respect to the optimal adaptive algorithm. The first algorithm uses a modified round-robin approach among three sequences, which is inspired by a recent result on the unit cost case of the SSClass problem~\cite{DBLP:journals/algorithmica/GrammelHKL22}. The second algorithm is originally from the work of Gkenosis et al. \cite{DBLP:conf/esa/GkenosisGHK18}. In our work, we successfully improve its approximation factor from $2(B-1)$ to $6$. Our analysis heavily uses the relation between computation and verification of functions, which was studied in the information theory literature \cite{DBLP:conf/allerton/AcharyaJO11,DBLP:conf/isit/DasJOPS12}.
\end{abstract}

\section{Introduction}

Consider the unweighted Stochastic Score Classification problem where we want to minimize the expected cost of evaluating a symmetric Pseudo-Boolean function $f(x_1,\dots,x_n):\{0,1\}^n\rightarrow \{1,\dots,B\}$ on an initially unknown input. The value of each $x_i$ is unknown and can only be learned by paying a positive cost $c_i$. We also know that each $x_i$ has an independent probability $p_i$ of being 1. The function $f$ contains $B+1$ integer intervals $0=\alpha_1<\alpha_2<\cdots<\alpha_{B+1}=n+1$ which are called the \emph{scoring intervals}. The value of function $f$ equals $j$ if the \emph{score} $\sum x_i$ satisfies $\alpha_j\leq \sum x_i \leq \alpha_{j+1}-1$. Testing continues until the value of the function becomes a foregone conclusion. Since the value of the function only depends on the number of 1's in its input, we say this function is symmetric.

The Stochastic Score Classification (SSClass) problem was generalized and studied by Gkenosis et al.~\cite{DBLP:conf/esa/GkenosisGHK18}, in which they gave both adaptive and non-adaptive algorithms to the weighted and unweighted versions of the problem.
Prior to that, the unit cost unweighted version was studied in the information theory literature~\cite{DBLP:conf/allerton/AcharyaJO11,DBLP:conf/isit/DasJOPS12} and the focus was to study the equivalence between verification and evaluation. Grammel et al.~\cite{DBLP:journals/algorithmica/GrammelHKL22} focused on the unit cost case of the unweighted SSClass problem and gave a 2-approximation non-adaptive algorithm. Gkenosis et al.~\cite{DBLP:journals/dam/GkenosisGHK22} also improved their results to give a $(B-1)$-approximation non-adaptive algorithm for the unweighted arbitrary cost SSClass problem.

When $B=2$, the unweighted version of the SSClass problem is equivalent to the \emph{Stochastic Boolean Function Evaluation} problem for $k$-of-$n$ functions where $f=1$ iff $\sum x_i\geq k$ and $0$ otherwise. The adaptive version of this problem has an elegant optimal algorithm that was discovered by Salloum, Breuer, and (independently) Ben-Dov~\cite{Salloum97,BenDov81}. For the non-adaptive version, Gkenosis et al. gave a 2-approximation algorithm using a simple round-robin strategy. Grammel et al. designed a 1.5-approximation algorithm for the unit cost case, which exactly matches the  adaptivity gap of the $k$-of-$n$ function. The Stochastic Boolean Function Evaluation (SBFE) problem, which was first studied by Deshpande et al., shares many common settings with the SSClass problem. The major difference is that, in the SBFE problem, the goal is to minimize the expected cost of evaluating a Boolean function rather than a Pseudo-Boolean function in the SSClass problem. Both the SBFE problem and the SSClass problem are NP-hard, since when $B=2$ the weighted SSClass problem corresponds to the SBFE problem for linear threshold functions, which is known to be NP-hard.

It was an open question whether there exists a polynomial-time constant-factor approximation algorithm for the general unweighted SSClass problem. Surprisingly, Ghuge et al.~\cite{DBLP:conf/ipco/GhugeGN22} studied the weighted version of the SSClass problem and provided the first polynomial-time constant-factor non-adaptive approximation algorithm. Their work closed the previous question, however, although Ghuge et al. did not try to minimize their approximation factor, it appears to be at least 100 in their setting. Hence, it would be interesting to figure out how to design non-adaptive approximation algorithms for the general unweighted SSClass problem with a reasonably small approximation factor.

In this work, we successfully provide two different 6-approximation non-adaptive algorithms. Our analysis of the two algorithms uses the properties of verification and evaluation. We first show that each of the two algorithms give a 6-approximation verification strategy for any $j\in\{1,\dots,B\}$, then since the strategy does not depend on any particular $j$, we conclude that both our algorithms are 6-approximation evaluation algorithms. Our results also indicate that the adaptivity gap for the unweighted SSClass problem is at most 6.

\section{Preliminaries}

For the SSClass problem, we can understand that the scoring intervals divide possible scores $\{0,\dots,n+1\}$ into multiple blocks. We say a realization $a\in\{0,1\}^n$ belong to a \emph{block} $j$ iff $f(a)=j$. Let $N_0(a), N_1(a)$ denote the number of 0's and 0's in realization $a$, respectively. From the definition we know that $f(a)=j$ iff $\alpha_j\leq N_1(a) \leq \alpha_{j+1}-1$. Since there are totally $n$ variables, the second inequality hints that $N_0(a)$ must satisfy $N_0(a)\geq n-(\alpha_{j+1}-1)$. Let $b_j=n-(\alpha_{j+1}-1)$, we can say assignment $a$ belongs to block $j$ if both $N_0(a)\geq b_j$ and $N_1(a)\geq a_j$ are satisfied, and we define $M_j$ such that
$$M_j:=\{a\mid a\in\{0,1\}^n, f(a)=j\}$$

Now we define the verification costs and evaluation costs of strategies. Let $C_j(S,a)$ denote the cost of running strategy $S$ to verify that realization $a$ belongs to block $j$. Then we define $E_j(S)$ to be the verification cost of block $j$ using strategy $S$:
$$E_j(S)=\sum_{a\in M_j}C_j(S,a)p(a)$$
If $S$ does not depend on any information of block $j$, we define $E(S)$ to be the evaluation cost of $f$ using strategy $S$:
$$E(S)=\sum_{j\in [B]}E_j(S)$$
Let $\mathcal{S}$ denote the set of all possible strategies. We define $\OPT_j$ to be an optimal verification strategy for block $j$:
$$E_j(\OPT_j)=\min_{S\in\mathcal{S}}E_j(S)$$
We also define $\OPT$ to be an optimal evaluation strategy for function $f$: 
$$E(\OPT)=\min_{S\in\mathcal{S}}E(S)$$

\paragraph{Adaptive strategies} A strategy is called an adaptive strategy if the choice of the next test to perform can depend on the results of previous tests. An adaptive strategy in our problem can be considered as a binary tree where each non-leaf node represents a variable to test and each arc of the node represent the 0/1 result of the variable. 

\paragraph{Non-adaptive strategies} A strategy is called a non-adaptive strategy if it cannot determine the next test to perform based on previous results; hence we may consider it as a permutation of variables. Testing continues until the value of $f$ is a foregone conclusion. If we represent the non-adaptive strategy as a binary tree, all nodes in the same level need to represent the same variable.

\paragraph{Adaptivity gap} The adaptivity gap captures the maximum difference between the optimal adaptive strategy and the optimal non-adaptive strategy. Intuitively, the adaptivity gap represents how ``bad'' the optimal non-adaptive strategy can be compared to the optimal adaptive strategy in the worst case; or how much benefit can be obtained if we are allowed to be adaptive in the (internal) decision-making process.

% \bigskip

% \noindent Define a realization $a\in\{0,1\}^n$ belong to block $i$ iff $N_1(a)\geq \alpha_i$ and $N_0(a)\geq\beta_i$, here $N_1(a), N_0(a)$ are the numbers of 1,0 in realization $a$. Let $M_i$ contains all possible realizations that belong to block $i$.

% Let $C_i(S,a)$ to note the cost of running strategy $S$ to verify that realization $a$ is in block $i$. Then we let $E_i(S)$ be the verification cost of block $i$ using strategy $S$,
% $$E_i(S)=\sum_{a\in M_i}C_i(S,a)p(a)$$
% and we let $E(S)$ be the evaluation cost on function $f$ using strategy $S$. Suppose $f$ have $B$ blocks,
% $$E(S)=\sum_{i\in [B]}E_i(S)$$
% Let $\mathcal{S}$ denote the set of all possible strategies. We define $\OPT_i$ as the optimal verification strategy for block $i$,
% $$\OPT_i=\argmin_{S\in\mathcal{S}}E_i(S)$$
% we also define $\OPT$ as the optimal evaluation strategy for function $f$,
% $$\OPT=\argmin_{S\in\mathcal{S}}E(S)$$

\section{Results}

\subsection{The 3-way modified round-robin algorithm}
Our first algorithm uses a modified round-robin approach which was previously introduced by \cite{allen2017evaluation}. The modified round-robin approach was designed to aggregate two sequences with heterogeneous costs. In our case, we extend it to aggregate three sequences with heterogeneous costs. 

Before we go to the 6-approximation proof of the evaluation strategy, we first propose a verification strategy $V_j$ for all assignments in any particular block $j$, using the following three sequences:\bigskip

\noindent\textbf{$P_0$}\ : all $x_i$ are sorted in the increasing order of $c_i/(1-p_i)$

\noindent\textbf{$P_1$}\ : all $x_i$ are sorted in the increasing order of $c_i/p_i$

\noindent\textbf{$P_c$}\ : all $x_i$ are sorted in the increasing order of $c_i$\bigskip

The verification strategy $V_j$ runs as follows

\begin{algorithm}
\caption{Verifying strategy $V_j$ for block $j$}\label{alg:V}
\begin{algorithmic}[1]
% \Require $n \geq 0$
% \Ensure $y = x^n$
\State $\pi \gets \text{first $\alpha_j+\beta_j$ variables in $P_c$}$
\State test all variables in $\pi$
% \State $\pi_2 \gets []$
\While{we have not find at least $\alpha_j$ 1's and $\beta_j$ 0's}
\If{we have found at least $\alpha_j$ 1's}
    \State $x \gets \text{the next variables in $P_0$ if it is not in $\pi$}$
\ElsIf{we have found at least $\beta_j$ 0's}
    \State $x \gets \text{the next variables in $P_1$ if it is not in $\pi$}$
\EndIf
\State test $x$
\EndWhile
\end{algorithmic}
\end{algorithm}

\begin{lemma}\label{lemma:2-approx-verification}
    Strategy $V_j$ is a 2-approximation verification strategy for verifying block $j$, which means
    $$E_j(V_j)\leq 2E_j(\OPT_j)$$
\end{lemma}

To prove Lemma~\ref{lemma:2-approx-verification}, we need to use the following lemmas. First, our original verification problem in Lemma~\ref{lemma:2-approx-verification} contains a cost vector $c=(c_1,\dots,c_n)$. Consider a new problem with a different cost vector $c'=(c_1',\dots,c_n')$, and let set $A$ contains the indexes of the first $\alpha_j+\beta_j$ variables in $P_c$. Then we define $c'$ as
\begin{align*}
    &\forall i\in A, c_i'=0\\
    &\forall i\in [n]\setminus A, c_i'=c_i
\end{align*}
let $E_j'(S)$ be the cost of verifying block $j$ using strategy $S$ with the new cost vector $c'$, and let $\OPT_j'$ to be the optimal verification strategy for block $j$ with the new cost vector $c'$
\begin{lemma}\label{lemma:0-cost}
    $$E_j'({\OPT_j'})\leq E_j({\OPT_j})$$
\end{lemma}
\begin{proof}
    Consider using $\OPT_j$ in our new problem with cost vector $c'$. Obviously
    $$E_j'(\OPT_j)\leq E_j(\OPT_j)$$
    because $\forall i\in[n], c_i'\leq c_i$. Since $\OPT_j'$ is the optimal verification strategy for block $j$ with $c'$,
    $$E_j'(\OPT_j')\leq E_j'(\OPT_j)$$
\end{proof}

\begin{lemma}\label{lemma:testing-0-cost}
    There exists an optimal verification strategy for the new problem which first tests all $x_i$ where $c_i'=0$.
\end{lemma}
\begin{proof}
    We first note that each strategy corresponds to a decision tree. Suppose there exists an optimal strategy that did not test a ``free'' variable $x'$ before testing any variables with costs. Then for any root-to-leaf path that contains $x'$, moving $x'$ to the root of the decision tree will not change the cost on that path; for any root-to-leaf path does not contain $x'$, first testing $x'$ will never increase the cost on that path as well. We can reuse this argument until all free variables are tested before testing any costly variables.
\end{proof}

\begin{lemma}\label{lemma:vi-optimal}
    $V_j$ is an optimal strategy for the new problem with cost vector $c'$.
\end{lemma}
\begin{proof}
    We can divide Algorithm~\ref{alg:V} into two phases. In the first phase, we test all 0-cost variables. If we have not yet verified a realization is in block $j$ after testing those variables, we are either short of 1's or short of 0's. For instance, suppose we are short of two 1's after phase 1; then we need to decide the optimal verification strategy for the remaining variables. Let $\pi$ contain the variables we tested in phase 1. We define the sub-problem as verifying that there exist at least two 1's in $\{x_1,\dots,x_n\}\setminus \pi$, given there are at least two 1's in $\{x_1,\dots,x_n\}\setminus \pi$. The second phase of Algorithm~\ref{alg:V} is optimal since we know that testing the variables in the increasing order of $c_i/p_i$ is the optimal strategy for verifying at least $k$ 1's in a $k$-of-$n$ function. Combining this argument with our previous Lemma~\ref{lemma:testing-0-cost} finishes the proof.
\end{proof}

\begin{proof}[Proof of Lemma~\ref{lemma:2-approx-verification}]
We can compare our Algorithm~\ref{alg:V} on the verification problem with $c$ to the verification problem with $c'$. The only difference is the costs incurred in phase 1. Hence
$$E_i(V_j)=E_j'(V_j)+\sum_{i:x_i\in\pi}c_i$$
Since we need to run at least $\alpha_j+\beta_j$ tests to verify any realization $a$ is in block $j$,
$$\sum_{i:x_i\in\pi}c_i\leq E_j(\OPT_j)$$
and from Lemma~\ref{lemma:vi-optimal} and Lemma~\ref{lemma:0-cost} we know $E_j'(V_j)=E_j'(\OPT_j')\leq E_j(\OPT_j)$,
$$E_j(V_j)=E_j'(V_j)+\sum_{i:x_i\in\pi}c_i\leq 2E_j(\OPT_j)$$
\end{proof}

We now present our (modified) round-robin evaluation algorithm for any symmetric function $f$.

\begin{algorithm}[H]
\caption{Round-robin strategy $RR$}\label{alg:RR}
\begin{algorithmic}[1]
% \Require $n \geq 0$
% \Ensure $y = x^n$
\State $K_0\gets 0, K_1\gets 0, K_c\gets 0$
\State $\pi \gets []$
\While{we have not find at least $\alpha_j$ 1's and $\beta_j$ 0's for any $j\in\{1,\dots,B\}$}
\State $x_0\gets \text{the next variable in $P_0$ not in $\pi$}$
\State $x_1\gets \text{the next variable in $P_1$ not in $\pi$}$
\State $x_c\gets \text{the next variable in $P_c$ not in $\pi$}$
\If{$K_0+c_0=\min\{K_0+c_0,K_1+c_1,K_c+c_c\}$}
    \State test $x_0$
    \State $\pi \gets \pi+x_0$
    \State $K_0 = K_0 + c_0$
\ElsIf{$K_1+c_1=\min\{K_0+c_0,K_1+c_1,K_c+c_c\}$}
    \State test $x_1$
    \State $\pi \gets \pi+x_1$
    \State $K_1 = K_1 + c_1$
\ElsIf{$K_c+c_c=\min\{K_0+c_0,K_1+c_1,K_c+c_c\}$}
    \State test $x_c$
    \State $\pi \gets \pi+x_c$
    \State $K_c = K_c + c_c$
\EndIf
\EndWhile\\
\Return $f(x_1,\dots,x_n)=j$
\end{algorithmic}
\end{algorithm}

\begin{lemma}\label{lemma:6-approx-verification}
    $RR$ is a 6-approximation verification strategy for verifying block $j$:
    $$E_j(RR)\leq 6E_j(\OPT_j)$$
\end{lemma}
\begin{proof}
    Recall we separate Algorithm~\ref{alg:V} into two phases: in phase 1 we test the cheapest $\alpha_j+\beta_j$ variables, and in phase 2 we test the remaining variables. Consider an arbitrary realization $a$, there can be three cases: 1) the algorithm terminates after phase 1; 2) the algorithm has not found enough 0's in phase 1; 3) the algorithm has not found enough 1's in phase 1. We let $\pi$ denote the variables tested in phase 1 and $\pi'$ be the variables tested in phase 2 when we run strategy $V_j$ on realization $a$.
    
    \textbf{Case 1}: Algorithm~\ref{alg:V} found both $\alpha_j$ 1's and $\beta_j$ 0's in phase 1.
    
   Let $\Bar{x}$ be the last variables tested by Algorithm~\ref{alg:RR} and it has cost $\Bar{c}$. Let $\sigma_0,\sigma_1,\sigma_c$ denote the tested variables on sequences $P_0,P_1,P_c$ right before testing $\Bar{x}$. Let  the next variable on $P_0$ be $x_{P_0}$, we define $\sigma_0^+=\sigma_0\cup\{x_{P_0}\}$. $x_{P_1},x_{P_c},\sigma_1^+,\sigma_c^+$ are defined similarly. Clearly, $\sigma_c\subset\pi$, because otherwise all variables in $\pi$ have been tested and we must have found enough 1's and 0's. 
   Also, it must be that $x_{P_c}\in\pi$. If not, suppose $x_{P_c}\notin\pi$; then all variables in $P_c$ and ahead of $x_{P_c}$ must have been tested (which includes all variables in $\pi$), and we know it is not true since the algorithm has not yet terminated. Thus we conclude $\sigma_c^+\subseteq \pi$.
    By the round-robin principle, \begin{align*}
        C_i(RR,a)&=\sum_{j:x_j\in\sigma_0}c_j+\sum_{j:x_j\in\sigma_1}c_j+\sum_{j:x_j\in\sigma_c}c_j+\Bar{c}\\
        &\leq 3\sum_{j:x_j\in \sigma_c^+}c_j\leq 3\sum_{j:x_j\in \pi}c_j=3C_i(V_i,a)
    \end{align*}
    
    \textbf{Case 2}: Algorithm~\ref{alg:V} found at least $\alpha_j$ 1's but fewer than $\beta_j$ 0's in phase 1.
    
    We separate Algorithm~\ref{alg:RR} into two phases as well. It starts in phase 1 and moves to phase 2 after it has found just $\alpha_j$ 1's. Let $C_j(RR^1,a),C_j(RR^2,a)$ denote the cost spent in phase 1 and phase 2 of Algorithm~\ref{alg:RR} respectively, and similarly, let $C_j(V_j^1,a),C_j(V_j^2,a)$ denote the cost spent in phase 1 and phase 2 of Algorithm~\ref{alg:V} respectively.

    If by the end of phase 1 Algorithm~\ref{alg:RR} has collected at least $\alpha_j$ 1's and at least $\beta_j$ 0's, we can reuse the previous argument to get
    $$C_j(RR,a)=C_j(RR^1,a)\leq 3C_j(V_j^1,a)\leq 3C_j(V_j,a)$$
    We now focus on the case that Algorithm~\ref{alg:RR} at least test one variable in phase 2. Since Algorithm~\ref{alg:RR} has found $\alpha_j$ 1's in phase 1, when it terminates in phase 2, it must have just found $\beta_j$ 0's. We know at least one of the two conditions $x_{P_c}\in\pi,x_{P_0}\in\pi'$ hold, otherwise, all variables in $\pi\cup\pi'$ must have been tested and the algorithm should have terminated before testing the last variable. If $x_{P_c}\in \pi$ holds, then $\sigma_c^+\subseteq \pi$. From the round-robin principle, we know
    $$C_j(RR,a)\leq 3\sum_{i:x_i\in\sigma_c^+}c_i\leq 3\sum_{i:x_i\in\pi}c_i= 3C_j(V_j^1,a)\leq 3C_j(V_j,a)$$
    If $x_{P_c}\in \pi$ does not hold, $x_{P_0}\in\pi'$ must hold, which gives us $\sigma_0^+\subseteq \pi'$. Then,
    $$C_j(RR,a)\leq 3\sum_{i:x_i\in\sigma_0^+}c_i\leq 3\sum_{i:x_i\in\pi'}c_i= 3C_j(V_j^2,a)\leq 3C_j(V_j,a)$$
    
    \textbf{Case 3}: Algorithm~\ref{alg:V} found at least $\beta_j$ 0's but less than $\alpha_j$ 1's in phase 1.
    
    Using a symmetric argument from case 2, $C_j(RR,a)\leq 3C_j(V_j,a)$ holds for case 3 as well.\bigskip
    
    Since $C_j(RR,a)\leq 3C_j(V_j,a)$ holds in all three possible cases, 
    \begin{align*}
        E_j(RR)=\sum_{a\in M_j}C_j(RR,a)p(a)\leq 3\sum_{a\in M_j}C_j(V_j,a)p(a)=3E_j(V_j)\leq 6E_j(\OPT_j)
    \end{align*}
\end{proof}
\begin{theorem}\label{theorem:1}
    $RR$ is a 6-approximation evaluation strategy.
\end{theorem}
\begin{proof}
    \begin{align*}
        E(RR)=\sum_{j\in [B]}E_j(RR)\leq 6\sum_{j\in [B]}E_j(\OPT_j)\leq 6\sum_{j\in [B]}E_j(\OPT) \leq 6E(\OPT)
    \end{align*}
Here the second inequality happens because the optimal verification cost is always upper-bounded by the optimal evaluation cost.
\end{proof}

\subsection{The 2-way modified round-robin algorithm}
Now we introduce the second algorithm, which uses the modified round-robin approach on two sequences. This algorithm is first proposed by Gkenosis et al.~\cite{DBLP:conf/esa/GkenosisGHK18} and they showed that the algorithm yields a $2(B-1)$-approximation. The algorithm is as follows
\begin{algorithm}[H]
\caption{Round-robin strategy $RR2$}\label{alg:RR2}
\begin{algorithmic}[1]
% \Require $n \geq 0$
% \Ensure $y = x^n$
\State $K_0\gets 0, K_1\gets 0$
\State $\pi \gets []$
\While{we have not find at least $\alpha_j$ 1's and $\beta_j$ 0's for any $j\in\{1,\dots,B\}$}
\State $x_0\gets \text{the next variable in $P_0$ not in $\pi$}$
\State $x_1\gets \text{the next variable in $P_1$ not in $\pi$}$
\If{$K_0+c_0\leq K_1+c_1$}
    \State test $x_0$
    \State $\pi \gets \pi+x_0$
    \State $K_0 = K_0 + c_0$
\Else
    \State test $x_1$
    \State $\pi \gets \pi+x_1$
    \State $K_1 = K_1 + c_1$
\EndIf
\EndWhile\\
\Return $f(x_1,\dots,x_n)=j$
\end{algorithmic}
\end{algorithm}
We propose a new proof which gives a 6-approximation. Our proof reuses many important properties in the proof of our previous algorithm. In addition, we need to use two extra lemmas.

Let $\Pi$ denote the first $k$ variables in $P_c$. Let $\Pi'$ be an empty set. We compare the first variables in both $P_0$ and $P_1$, then add the variable with a smaller cost to $\Pi'$ and remove that variable from both $P_0$ and $P_1$.\footnote{If the two variables have the same costs, we break ties arbitrarily.} We repeat this procedure for $k$ times so $\Pi'$ has $k$ variables. At last, we let $\Pi''$ denote the first $k$ variables output by the modified round-robin strategy $RR2$ if we ignore the stopping condition.
\begin{lemma}\label{lemma:6}
    $$\sum_{i:x_i\in\Pi'}c_i\leq 2\sum_{i:x_i\in\Pi}c_i$$
\end{lemma}
\begin{proof}
    Let $x_0, x_1$ denote the first variable in the remaining sequence $P_0, P_1$, and let their costs be $c_0, c_1$, respectively. Define $x^*$ to be the variable that has the smallest testing cost among all remaining variables and let its cost be $c^*$. We first want to prove
    $$\min\{c_0,c_1\}\leq 2c^*$$
    Since $x_0$ is the first variable in $P_0$, we have
    $$\frac{c_0}{1-p_0}\leq \frac{c^*}{1-p^*} \Rightarrow \frac{1-p_0}{c_0} \geq \frac{1-p^*}{c^*}$$
    and since $x_1$ is the first variable in $P_1$, we have
    $$\frac{c_1}{p_1}\leq \frac{c^*}{p^*} \Rightarrow \frac{p_1}{c_1} \geq \frac{p^*}{c^*}$$
    Hence,
    \begin{align*}
        \frac{1-p_0}{c_0}+\frac{p_1}{c_1} \geq \frac{1}{c^*}&\Rightarrow
        \frac{1-p_0+p_1}{\min\{c_0,c_1\}} \geq \frac{1}{c^*}\\
        &\Rightarrow
        \min\{c_0,c_1\} \leq (1-p_0+p_1)\cdot c^*\leq 2c^*
    \end{align*}
    Since we can consider that all variables in $\Pi'$ are inserted sequentially, then each time it needs to compare the two variables $x_0, x_1$ with the variable which has the smallest cost among the remaining variables. Hence, we can see for all $i\in[k], \Pi_i'\leq 2\Pi_i$, the lemma follows.
\end{proof}
\begin{lemma}\label{lemma:7}
    $$\sum_{i:x_i\in\Pi''}c_i\leq 2\sum_{i:x_i\in\Pi'}c_i$$
\end{lemma}
\begin{proof}
    Let the last variable added to $\Pi''$ be $\bar{x}$ and without loss of generality assume $\bar{x}$ is selected from $P_0$. Let the first remaining variable on the other sequence $P_1$ be $\hat{x}$. Let $C_0, C_1$ denote the accumulated costs on $P_0, P_1$ before adding $\bar{x}$, respectively.
    
    Since $\bar{x}$ is the last variable to be added, $C_1\leq C_0+c_{\bar{x}}\leq C_1+ c_{\hat{x}}$.
    \paragraph{Case 1:} Suppose $\bar{x}\in \Pi'$, $C_0+c_{\bar{x}}\leq \sum_{i:x_i\in\Pi'}c_i$. Hence
    $$C_0+C_1+c_{\bar{x}}\leq 2C_0+c_{\bar{x}}\leq 2 \sum_{i:x_i\in\Pi'}c_i$$
    \paragraph{Case 2:} Suppose $\bar{x}\notin \Pi'$, $\hat{x}$ must in $\Pi'$, and $C_1+c_{\hat{x}}\leq \sum_{i:x_i\in\Pi'}c_i$. Then
    $$C_0+C_1+c_{\bar{x}}\leq 2C_1+c_{\hat{x}}\leq 2\sum_{i:x_i\in\Pi'}c_i$$
\end{proof}

Combining Lemma~\ref{lemma:6} and Lemma~\ref{lemma:7}, we have
\begin{corollary}\label{4times}
    $$\sum_{i:x_i\in\Pi''}c_i\leq 4\sum_{i:x_i\in\Pi}c_i$$
\end{corollary}

Now we are ready to prove our Theorem~\ref{thm:2}. In the proof, we will use a similar method as the proof of our Theorem~\ref{theorem:1}, and we will skip some duplicated details.
\begin{theorem}\label{thm:2}
    $RR2$ is a 6-approximation evaluation strategy.
\end{theorem}
\begin{proof}
    For the analysis purpose, we first introduce an adaptive verification strategy $V'$ for any block $j$. The strategy uses the $RR2$ algorithm to get $\alpha_j+\beta_j$ variables, which we call phase 1. Then, if we have not yet verified that an assignment is in block $j$, we must either lack 1's or lack 0's, and can not lack both. In phase 2, if we lack 1's, we test the remaining variables in the increasing order of $\frac{c_i}{p_i}$; if we lack 0's, we test the remaining variables in the increasing order of $\frac{c_i}{1-p_i}$.
    
    We can use a similar argument in the proof of Lemma~\ref{lemma:6-approx-verification} to show that this adaptive verification strategy yields a 5-approximation with respect to the optimal verification strategy for any block $j$. The reason is that, any verification strategy for block $j$ must test at least $\alpha_j+\beta_j$ variables. Since we test $\alpha_j+\beta_j$ variables in phase 1 using $RR2$, we can see the cost we spend in phase 1 is upper-bounded by 4 times the optimal verification cost because of Corollary~\ref{4times}. Then, using a similar argument of Lemma~\ref{lemma:0-cost}, the cost we spend in phase 2 is upper-bounded by the optimal verification cost. Hence, the verification strategy $V'$ yields a 5-approximation.
    
    Now we can think that, in phase 2, we use a modified round-robin approach to test the remaining variables. Also, it is crucial that the modified round-robin approach ``inherit'' the accumulated costs in phase 1. Hence, this modification gives us the exact Algorithm~\ref{alg:RR2}. 
    
    % We can observe that the variables we test in phase 1 are fixed.
    It is important to observe that the cost of our strategy is spent either on $P_0$ or on $P_1$, which we can analyze separately.
    Since the variables which $RR2$ selects in phase 1 are fixed, we let $\pi_0, \pi_1$ denote the variables that are from $P_0, P_1$, respectively.
    Without loss of generality, suppose for an assignment $a$ we still need some 1's after phase 1. Strategy $V'$ will then test variables using sequence $P_1$. Let $\pi$ denote the variables strategy $V'$ tests in phase 2 on $P_1$; then the cost of $V'$ on $a$ equals $\sum_{i:x_i\in \pi_0\cup\pi_1\cup\pi}c_i$. 
    If we run $RR2$ on this realization $a$, 
    when the test ends, the cost spent on $P_1$ is at most $\sum_{i:x_i\in \pi_1\cup \pi}$, because otherwise all tests in $\pi$ have been tested and the testing process should have ended. Also, from the modified round-robin strategy we also know the cost spent on $P_0$ is also at most $\sum_{i:x_i\in \pi_1\cup \pi}$. Hence, for any realizations, $RR2$ spends at most twice the cost in phase 2 than $V'$. Since the expected cost of $V'$ spent in phase 2 is upper-bounded by the optimal verification cost, the expected cost of phase 2 in $RR2$ is upper-bounded by twice the optimal verification cost. Combing with the cost spent in phase 1, Algorithm~\ref{alg:RR2} is a 6-approximation verification strategy for any block $j$. Since Algorithm~\ref{alg:RR2} does not depend on $\alpha_j$ and $\beta_j$, it is also a 6-approximation evaluation strategy.
\end{proof}

\section{Discussion}

We showed two constant-approximation algorithms for the unweighted SSClass problem with arbitrary cost. In our proof, we first focus on the verification strategies for any block $j$. Since our strategies do not use any information of block $j$ (e.g., $\alpha_j, \beta_j$), we could then use the same strategy for all $B$ blocks, hence the approximation factor we achieved for the verification problem can be extended to the evaluation problem.

Although our strategies suggest that it is possible to obtain strategies with small approximation factors for the unweighted SSClass problem, we do not know if we can achieve the same approximation factor for the arbitrary cost case and the unit cost case. In fact, the two cases have some crucial differences which prevent us from applying similar techniques in proof of the unit cost case to the arbitrary cost case.
For example, when all variables have unit testing costs, the optimal strategy of verifying at least $\alpha_j$ of 1's for all assignments in block $j$ is to test variables in the increasing order of $1-p_i$. However, we do not know the optimal strategy for the same problem in the arbitrary cost case; a trivial extension of testing in the increasing order of $\frac{c_i}{p_i}$ is not optimal, which can be proved by counter-examples. It would be interesting to see if one could further improve this approximation factor for the arbitrary cost case of the problem to, e.g., 2, which is obtained for the unit cost case unweighted SSClass.

Another interesting problem is to show the complexity of the unweighted SSClass problem. It is known that the weighted version is NP-hard since the special case when $B=2$ corresponds to the SBFE problem for linear threshold formulas which is NP-hard. However, it is open whether finding an optimal adaptive/non-adaptive strategy for the unweighted SSClass is NP-hard, even in the unit cost case.

\bibliographystyle{alpha}
\bibliography{IEEEabrv,bibliography}
\end{document}